\documentclass[11pt]{article}
\usepackage[margin=1in]{geometry}                % See geometry.pdf to learn the layout options. There are lots.
\usepackage[]{authblk}
\usepackage{graphicx}
\usepackage{amssymb}
\usepackage{amsmath}
\usepackage[ruled,vlined,linesnumbered]{algorithm2e}
\usepackage[usenames,dvipsnames,svgnames,table]{xcolor}
\usepackage{bm}

\renewcommand{\approx}[1]{{\underline{#1}}}

\newtheorem{lemma}{Lemma}
\newtheorem{theorem}{Theorem}

\newenvironment{proof}{\trivlist\item[\hskip\labelsep{\bf Proof}:]}{\hfill
 $\Box$ \endtrivlist}
\newenvironment{proofSpec}{\trivlist\item[\hskip\labelsep{\bf Proof of Theorem~\ref{thm:fptasDAGs} for the DAG case}:]}{\hfill
 $\Box$ \endtrivlist}
\newenvironment{proofSpecLemma}{\trivlist\item[\hskip\labelsep{\bf Proof of Lemma~\ref{lemma:list'}}:]}{\hfill
 $\Box$ \endtrivlist}

\newcommand{\eps}{\varepsilon}

\newcommand{\fl}[1]{\left\langle{#1}\right\rangle}
\newcommand{\inset}[2]{\in \{{#1},\dots,{#2}\}}
\title{\vspace{-1.3cm}Combinatorial decomposition approaches for efficient counting and random generation FPTASes}
\date{}                                           % Activate to display a given date or no date

\author{Romeo Rizzi$^{1}$, Alexandru I. Tomescu$^{2}$}
\affil{\normalsize $^{1}$Department of Computer Science, University of Verona, Italy\\
$^{2}$Helsinki Institute for Information Technology HIIT,\\Department of Computer Science, University of Helsinki, Finland\\
\texttt{romeo.rizzi@univr.it}, \texttt{tomescu@cs.helsinki.fi}}

\begin{document}
\maketitle
\vspace{-1cm}
\begin{abstract}
Given a combinatorial decomposition for a counting problem, we resort to the simple scheme of approximating large numbers by floating-point representations in order to obtain efficient Fully Polynomial Time Approximation Schemes (FPTASes) for it. The number of bits employed for the exponent and the mantissa will depend on the error parameter $0 < \varepsilon \leq 1$ and on the characteristics of the problem. Accordingly, we propose the first FPTASes with $1 \pm \varepsilon$ relative error for counting and generating uniformly at random a labeled DAG with a given number of vertices. This is accomplished starting from a classical recurrence for counting DAGs, whose values we approximate by floating-point numbers. 

After extending these results to other families of DAGs, we show how the same approach works also with problems where we are given a compact representation of a combinatorial ensemble and we are asked to count and sample elements from it. We employ here the floating-point approximation method to transform the classic pseudo-polynomial algorithm for counting 0/1 Knapsack solutions into a very simple FPTAS with $1 - \varepsilon$ relative error. Its complexity improves upon the recent result (\v{S}tefankovi\v{c} et al., SIAM J. Comput., 2012), and, when $\varepsilon^{-1} = \Omega(n)$, also upon the best-known randomized algorithm (Dyer, STOC, 2003). To show the versatility of this technique, we also apply it to a recent generalization of the problem of counting 0/1 Knapsack solutions in an arc-weighted DAG, obtaining a faster and simpler FPTAS than the existing one.

% directed acyclic graph, DAG, 0/1 Knapsack, Knapsack on DAG, approximate counting, random generation, FPTAS, floating-point arithmetic, combinatorial decomposition, counting recurrence, dynamic programming, deterministic algorithm

\end{abstract}

%\newpage

\section{Introduction}

In this paper we consider two main types of counting problems. In the first (combinatorial family), the input consists of a single integer $n$ and we are interested in counting/generating the objects of the $n$th slice of a family parametrized by $n$, such as all labelled trees on $n$ vertices or all well-formed formulas on $n$ parentheses; in this paper we tackle labeled directed acyclic graphs (DAGs) on $n$ vertices, and two DAG subclasses. In the second (combinatorial ensemble), we are given a structure and we want to count and sample all of its substructures with a given property, such as the spanning trees or the perfect matchings of a graph given in input; we tackle here the problem of counting 0/1 Knapsack solutions, and a generalization of this problem to a DAG.

%We will consider problems of both types: for the former case, we tackle the counting and random generation problems of labeled directed acyclic graphs (DAGs), and DAG subclasses, and for the latter case, we focus on the problem of counting the number of 0/1 Knapsack solutions. Once the combinatorial structure of the problem has been understood, our floating-point approximation scheme can be applied on top of the counting recurrences derived from the decomposition of the problem. Accordingly, we first show that we can transform an exact polynomial-time algorithm for generating u.a.r.~labeled DAGs into an FPTAS,  gaining both in space and time. We then prove that the classic pseudo-polynomial algorithm for counting the number of 0/1 Knapsack solutions can be converted into a very simple FPTAS whose complexity bound improves upon the one in the recent result~\cite{Stefankovic:2012:DPA:2339941.2339944}, and, when $\eps^{-1} = \Omega(n)$, also upon~\cite{DBLP:conf/stoc/Dyer03}.

%:
A general result by Jerrum, Valiant and Vazirani~\cite{DBLP:journals/tcs/JerrumVV86} is that the problem of exact random uniform generation of `efficiently verifiable' combinatorial structures is reducible to the counting problem. Since in many cases the counting problem is either hard, or simply expensive in practice,~\cite{DBLP:journals/tcs/JerrumVV86} also shows that for self-reducible problems almost uniform random generation and randomized approximate counting are inter-reducible. This fact has determined that randomized approximate approaches, in particular based on Markov chains, have attracted most attention. Indeed, in spite of the fact that these approaches are general, in the sense that they do not rely on problem specific combinatorial decompositions, they allow for faster algorithms when approximate solutions are good enough.

The intended message of the present work is that the compromise towards approximate solutions can also take place in the context of methods based on combinatorial decompositions. In order to facilitate this, we build up a minimalistic layer of floating-point arithmetic suitably tailored for this purpose. Our idea dates back to Denise and Zimmermann~\cite{DBLP:journals/tcs/DeniseZ99}, which considered floating-point arithmetic for uniform random generation of decomposable structures. A `decomposable structure' is a combinatorial structure definable in terms of the `standard constructions' of disjoint union, Cartesian product, Sequence, Cycle, and Set (see e.g.~\cite{AnalyticCombinatoricsBook} for details). These include, for example, unordered or ordered trees, permutations, set partitions, but do not include more complex objects such as DAGs, nor substructures of a structure given in input, such as solutions of a given 0/1 Knapsack instance. Moreover, even though~\cite{DBLP:journals/tcs/DeniseZ99} relates the relative error with the length of the mantissa, their results are not stated in terms of FPTASes. An FPTAS is a deterministic algorithm that estimates the exact solution within relative error $1 \pm \eps$, in time polynomial in the input size and in $1/\eps$.

Given an error $0 < \eps \leq 1$, we represent large integers as floating-point numbers having an exact exponent, so that no overflow can occur, and a mantissa whose length depends on $\eps$ and is only as long as to guarantee a $1 \pm \eps$ relative error.

We show that floating-point arithmetic can be added as a technical layer on top of any suitable combinatorial decomposition of the problem at hand, obtaining efficient state-of-the-art, both deterministic FPTASes for counting, and practical random generation algorithms with explicit error probability bounds. Some of our FPTASes are actually linear in $\log(1/\eps)$, which, in the case of counting problems, means a linear dependence on the length of the output.

Until now, for all the problems considered in this paper, Monte Carlo algorithms had the best running times, with a performance guarantee either proven (like for counting 0/1 Knapsack solutions~\cite{DBLP:conf/stoc/Dyer03}) or just generally accepted (like for DAGs generation~\cite{DBLP:journals/endm/MelanconDB01}). Recently, other authors have proposed deterministic algorithms for approximate counting~\cite{Stefankovic:2012:DPA:2339941.2339944,DBLP:conf/focs/GopalanKMSVV11}. However, our deterministic algorithms are the first ones to reach and even improve upon the running times of the Monte Carlo algorithms. Considered also that we get rid of the error probability $\delta > 0$, it is quite remarkable that we close the gap between deterministic and randomized algorithms.

In the same way that Markov chains offer a fascinating layer of reusable theory, our approach is also unifying, with the required math for bounding the run-time in terms of $\eps$ embodied in the technical floating-point arithmetic layer. Even though its level of generality is not comparable, it still offers a conceptual tool that can guide and inspire the design of new algorithms. In this new scenario, the length of the mantissa becomes a resource, and minimizing its consumption leads one to reduce the number of subsequent approximation phases in the processing of the data flow.
This view indeed supported us in gaining an $n$ extra factor in Thm.~\ref{thm:kDAG-main}. Moreover, the algorithms inspired by this framework do not require the difficult ad-hoc analysis of rapidly mixing properties of Markov chains, necessary for a conclusive word on the actual computational complexity of a given problem.

Based on these facts, we hope to see a renewal of interest on methods grounded on the combinatorial decomposition of the problem at hand both in practical and theoretical studies on counting and random generation, where the problems allow.

%It has been indicated by previous works~\cite{DBLP:journals/toms/SalvyZ94,DBLP:journals/tcs/DeniseZ99,Bodirsky:2006:GOG:1127704.1127709} that the compromise towards approximate solutions can also take place in the context of methods based on combinatorial decompositions. As in~\cite{DBLP:journals/tcs/DeniseZ99}, we will achieve this by building up a general layer
%% for this make-up 
%in the language of a minimal bulk of floating-point arithmetic considerations suitably tailored for this purpose. It is only in the present work however that decompositions approaches get their revenge on Markov chains. At the same time, our results 
%%sparigliano le carte on
%show 
%the dichotomy between deterministic and Monte Carlo algorithms.

%The message of our paper is that  This will be achieved by building up a general layer
%% for this make-up 
%in the language of a minimal bulk of floating-point arithmetic considerations suitably tailored for this purpose. Not only this approach is a deterministic one, but it can also deliver better results. 

%Accordingly, we show that given $0 < \eps \leq 1$, we do not have to store and compute the tables of these counting recurrences in an exact manner, but we propose to simply use floating-point numbers of a smaller size, depending on $\eps$, size which guarantees that the probability of generating an object is within relative error $1\pm\eps$ from the uniform one.

\subsection{Counting and random generation for a combinatorial family}

To illustrate the floating-point approximation scheme for a combinatorial family, we focus on DAGs. They constitute a basic class of graphs, with applications in various fields. As in the case of other combinatorial objects, the problem of generating uniformly at random (u.a.r., for short) a DAG with $n$ labeled vertices was first tackled with a Markov Chain algorithm~\cite{DBLP:journals/endm/MelanconDB01,DBLP:journals/ipl/MelanconP04}. 
%From this point of view, it is quite surprising that the generation problem for them was studied so late.
The main issue behind such a randomized approach lies in the difficulty of proving the rapidly mixing property.
%, namely that the number of steps the Markov Chain needs to run in order to have uniform sampling is bounded by a polynomial in $n$.
This was the case here for DAGs, as such a proof never appeared. 
Steinsky~\cite{DBLP:journals/soco/Steinsky03} proposed a nice generalization of Pr\"ufer's encoding of labeled trees to labeled DAGs, and put forth ranking and unranking algorithms. These led to a deterministic random generation algorithm working in time $O(n^5)$ and space $O(n^5M(n^2))$ bits, where $M(t)$ is the slowdown factor of multiplying two $t$-bit numbers. 

Our solution is based on the decomposition of DAGs by sources, initially proposed by Robinson~\cite{R73} to obtain a counting recurrence of labeled DAGs with $n$ vertices and a given number of sources. We exploit this decomposition by generating a labeled DAG recursively, at each step generating its sources (and their out-going arcs) by using the values of the counting recurrence as probability distribution. To further illustrate this method, in Appendix~\ref{appendix:DAGs} we consider two recently studied subclasses of DAGs, \emph{essential} DAGs (essDAGs)~\cite{Andersson97}, and \emph{extensional} DAGs (extDAGs)~\cite{MT10}.

\begin{theorem}
A labeled DAG, an essential DAG, or extensional DAG, with $n$ vertices can be generated u.a.r.~in time $O(n^3)$, provided a table of size $O(n^4)$ bits, computable in time $O(n^5M(n))$, is available.
\label{thm:exactDAGs}
\end{theorem}

We then show that, instead of storing the values of the counting recurrence as exact numbers on $O(n^2)$ bits, we can store approximate floating-point numbers with $O(\log n)$ bits for the exponent, and $O(\log(n/\eps))$ bits for the mantissa. This leads to the first deterministic FPTASes for counting and random generation, as stated in the following theorems, where $a(n)$, $d(n)$, $e(n)$ denote the number of labeled DAGs, essential DAGs, and extensional DAGs, respectively, on $n$ vertices.

\begin{theorem}
For any $n \geq 1$, and for every $0 < \eps \leq 1$, we can compute an $O(\log(n/\eps))$-bit number $Z$ such that $(1-\eps)a(n) \leq Z \leq a(n)$, $(1-\eps)d(n) \leq Z \leq d(n)$, or $(1-\eps)e(n) \leq Z \leq e(n)$, in time $O(n^3\log(n/\eps)M(\log (n/\eps)))$.
\label{thm:counting-DAGs}
\end{theorem}

\begin{theorem}
\begin{sloppypar}
For any $n \geq 1$, and for every $0 < \eps \leq 1$, we can generate at random a labeled DAG, essential DAG, or extensional DAG, $D$ on $n$ vertices with probability $\approx{p}(D)$ such that $1-\eps \leq \approx{p}(D)a(n) \leq 1+\eps$, $1-\eps \leq \approx{p}(D)d(n) \leq 1+\eps$, or $1-\eps \leq \approx{p}(D)e(n) \leq 1+\eps$. This can be done in time $O(n^2 + n\log (n/\eps))$, provided a table of size $O(n^2\log (n/\eps))$ bits, computable in time $O(n^3\log (n/\eps)M(\log (n/\eps)))$, is available.
\end{sloppypar}
\label{thm:fptasDAGs}
\end{theorem}

Notice how, since $a(n)$, $d(n)$ and $e(n)$ are less than $2^{n^2}$, then $\varepsilon = 1/2^{n^2}$ implies full precision, and, in the case of Thm.~\ref{thm:fptasDAGs}, we get the same running times as in Thm~\ref{thm:exactDAGs}.

%In Appendix~\ref{appendix:DAGs}, we further motivate our case with two recently studied subclasses of DAGs, \emph{essential} DAGs (essDAGs)~\cite{Andersson97}, and \emph{extensional} DAGs (extDAGs)~\cite{MT10}, obtaining exact random generation algorithms and FPTASes with the same complexity bounds.

\subsection{Counting a combinatorial ensemble}

To illustrate the floating-point approximation scheme for a combinatorial ensemble, we choose the well-known problem of counting 0/1 Knapsack solution. We are given a set of $n$ nonnegative integer weights $w_1,\dots,w_n$ and an integer $C$ and are asked how many subsets of elements of $w_1,\dots,w_n$ sum up to at most $C$. Since this problem is \#P-complete, research has focused on approximation algorithms. The first one was a randomized subexponential time algorithm~\cite{DBLP:journals/cpc/DyerFKKPV93} based on near-uniform sampling of feasible solutions by a random walk. A rapidly mixing Markov chain appeared in~\cite{DBLP:journals/siamcomp/MorrisS04}, which provided the first Fully Polynomial Time Randomized Approximation Scheme (FPRAS), for this problem, that had remained open for a some time. An FPRAS with a complexity $O(n^3 + n^2\eps^{-2})$ was given in~\cite{DBLP:conf/stoc/Dyer03}, by combining dynamic programming and rejection sampling. This complexity bound can be improved to $O(n^{2.5}\sqrt{\log(\eps^{-1})} + n^2\eps^{-2})$ by a more sophisticated approach using randomized rounding~\cite{DBLP:conf/stoc/Dyer03}. Recently,~\cite{Stefankovic:2012:DPA:2339941.2339944,DBLP:conf/focs/GopalanKMSVV11} gave the first deterministic FPTAS for this problem, running in time $O(n^3\eps^{-1}\log(n/\eps))$. A weaker result, namely a version of the algorithm of~\cite{Stefankovic:2012:DPA:2339941.2339944}, but in which the number of arithmetic operations depends on $\log C$, appeared in~\cite{DBLP:journals/corr/abs-1008-3187} and in the combined extended abstract~\cite{DBLP:conf/focs/GopalanKMSVV11}.

The solution in~\cite{Stefankovic:2012:DPA:2339941.2339944} is based on a function $\tau(i,a)$ defined as the smallest capacity $c$ such that there exist at least $a$ solutions to the 0/1 Knapsack problem with weights $\{w_1,\dots,w_i\}$ and capacity $c$. The second parameter of $\tau$ is then approximated, and $\tau(i,a)$ is computed by a dynamic programming algorithm. 

We start from the classic pseudo-polynomial dynamic programming algorithm obtained from the recurrence 
\[s(i,c) = s(i-1,c) + s(i-1,c-w_{i}),\]
where $s(i,c)$ is the number of 0/1 Knapsack solutions that use a subset of the items $\{w_1,\dots, w_i\}$, and their weights sum up to at most $c \leq C$. We approximate the values of $s$ using floating-point numbers, which leads to a more direct FPTAS, with a much simpler proof, and an easily implementable algorithm. Making the same assumption as~\cite{Stefankovic:2012:DPA:2339941.2339944} that additions on $O(\log C)$-bit numbers take unit time, we improve~\cite{Stefankovic:2012:DPA:2339941.2339944,DBLP:conf/focs/GopalanKMSVV11} as follows:

\begin{theorem}
For every $n \geq 1$, and every $0 < \eps \leq 1$, for an input $\{w_1,\dots,w_n\}$, $C$ to the 0/1 Knapsack counting problem, we can compute a floating-point number $Z$ of $2\log n + \log(1/\eps) + 1$ bits, which satisfies $(1-\eps)s(n,C) \leq Z \leq s(n,C)$,
in time $O(n^3\eps^{-1}\lceil\log (1/\eps)/\log n\rceil)$, assuming unit cost additions and comparisons on numbers with $O(\log C)$ bits.
\label{thm:knapsack-main}
\end{theorem}

\noindent Note that if $\eps^{-1} = \Omega(n)$, our deterministic FPTAS also improves both FPRASes in~\cite{DBLP:conf/stoc/Dyer03}. 

Our reasoning is along the following lines. Since the numbers of solutions can be at most $2^n$, and the values of the dynamic programming are obtained by sequences of $O(n)$ successive additions, we can approximate them using floating-point numbers with $\log n$ bits for the exponent and $\log (n/\eps) + 1$ bits for the mantissa. In order to obtain the $1 - \eps$ approximation factor, we will show that the relative error of each approximation of $s(i,c)$ is $(1-\eps/n)^i$, for any $i \leq n$. To keep the table small, we exploit the fact that the number of different entries in each row $i$ of the approximated table is at most $2^{\log n + \log (n/\eps)+1} = O(n^2/\eps)$. 

Recently, the problem of counting 0/1 Knapsack solutions has been extended to a DAG, as follows~\cite{WAOA13}. Given a DAG with nonnegative arc weights, two vertices $s$ and $t$, and a capacity $C$, count how many paths exist from $s$ to $t$ of total weight at most $C$; this problem is relevant for various applications in biological sequence analysis, see the references in~\cite{WAOA13}. This is clearly a generalization of counting 0/1 Knapsack solutions, since given an instance $\{w_1,\dots,w_n\},C$ it suffices to construct the DAG having $\{v_0,\dots,v_n\}$ as vertex set, $s = v_0$, $t=v_n$, and for each $i \inset{1}{n}$, there are two parallel arcs from $v_{i-1}$ to $v_i$, with weights 0 and $w_i$, respectively. 

\begin{sloppypar}
In~\cite{WAOA13}, the technique of~\cite{Stefankovic:2012:DPA:2339941.2339944} is extended to this problem, and an FPTAS running in time $O(mn^3\log n\eps^{-1}\log(n/\eps))$ is obtained (inaccurately, the factor $\log(n/\eps)$ is missing from their stated complexity bound). Just as we do for the classical 0/1 Knapsack problem, we start from the basic pseudo-polynomial dynamic programming algorithm extended to a DAG, whose values we approximate using floating-point numbers. We show that we can organize the computation in sequences of $O(n\log(\frac{m}{n}))$ successive additions, so that we need floating-point numbers with only $\log (n\log(\frac{m}{n})/\eps)$ bits for the mantissa, and $\log n$ bits for the exponent. This analogously leads to a faster and simpler FPTAS.% (The full details, very analogous to the proof of Thm.~\ref{thm:knapsack-main}, are in Appendix~\ref{appendix:kDAG}.)
\end{sloppypar}

\begin{theorem}
\begin{sloppypar}
For every $n \geq 1$, and every $0 < \eps \leq 1$, for an input DAG on $n$ vertices and $m$ arcs, nonnegative arc weights, and a capacity $C$, we can compute an $1-\eps$ approximation of the number of $s,t$-paths, in time $O(mn^2\log(\frac{m}{n})\eps^{-1}\lceil\log(1/\eps)/\log n\rceil)$, assuming unit cost additions and comparisons on numbers with $O(\log C)$ bits.
\end{sloppypar}
\label{thm:kDAG-main}
\end{theorem}

\section{Approximation by floating-point numbers}

Throughout this paper, we assume that the problem instances consist of $n$ objects (DAGs with $n$ vertices, 0/1 Knapsack instances with $n$ objects). Let $c \geq 1$ be such that the maximum numerical value of a particular counting problem is $2^{n^c}-1$ (that is, it can be represented with $n^c$ bits). Any number $x \in \{0,\dots,2^{n^c}-1\}$ can be written as 
\[x = x_{1} 2^{p-1} + x_{2} 2^{p-2} + \dots + x_{p-1} 2^1 + x_{p} 2^0 = 2^p\left(x_{1} 2^{-1} + x_{2} 2^{-2} + \dots + x_p 2^{-p}\right),\]
where $1 \leq p \leq n^c$, $x_1 = 1$, and $x_i \in \{0,1\}$, for $i \in \{2,\dots,p\}$. Under floating-point arithmetic terminology, $p$ is called the \emph{exponent} of $x$, and the binary string $x_1x_2\dots x_p$ is called its \emph{mantissa}.

We will approximate $x$ as a floating-point number which has $c \log n$ bits dedicated to store its exponent $p$ exactly, but only $t$ bits dedicated to store the first $t$ bits of its mantissa; that is, we approximate $x$ by the number
\[\fl{x}_{c\log n,t} := 2^p\left(x_{1} 2^{-1} + x_{2} 2^{-2} + \dots + x_t 2^{-t}\right).\]
We will often drop the subscript $c\log n,t$ when this will be clear from the context. For sure, we will choose $t \geq c\log n$, since the contrary cannot help. 

For every $0 \leq x < 2^{n^c}$, it holds that 
\begin{equation}
(1-2^{1-t})x \leq \fl{x}_{c\log n,t} \leq x.
\label{eq-1}
\end{equation}

Let $\approx{x}$ and $\approx{y}$ be two floating-point numbers with $c\log n$ bits for the exponent and $t$ bits for the mantissa. We denote the sum $\fl{\approx{x} + \approx{y}}$ by $\approx{x} \oplus \approx{y}$, and the product $\fl{\approx{x}\approx{y}}$ by $\approx{x} \otimes \approx{y}$. We assume that we can compute $\approx{x} \oplus \approx{y}$ with a bit complexity of $O(c \log n + t) = O(t)$; if additions on $O(\log n)$-bit numbers take unit time, then we assume we can compute $\approx{x} \oplus \approx{y}$ with a word complexity of $O(t/\log n)$. Let us denote by $M(t)$ the slowdown factor of a multiplication algorithm on $t$-bit numbers; for example, the Sch{\"o}nhage-Strassen algorithm~\cite{multiplication} multiplies two $t$-bit numbers in time $O(t\log t\log\log t)$, and we get $M(t) = \log t\log\log t$. Accordingly, we assume that we can compute $\approx{x} \otimes \approx{y}$ with $O((c \log n + t)M(c \log n + t))$ bit operations.

If $x,y \in \{0,\dots,2^{n^c}-1\}$ are such that $x + y,xy \in \{0,\dots,2^{n^c}-1\}$, and $\approx{x},\approx{y}$ are two floating-point numbers with $c\log n$ bits for the exponent and $t$ bits for the mantissa such that 
\[(1-2^{1-t})^i x \leq \approx{x} \leq x, \text{ and } (1-2^{1-t})^j y \leq \approx{y} \leq y,\]
\noindent for some integers $i,j \geq 0$, then by (\ref{eq-1}) the following inequalities hold

\begin{equation}
(1-2^{1-t})^{1+\max(i,j)}(x+y) \leq \approx{x} \oplus \approx{y} \leq x+y,
\label{eq:fl1}
\end{equation}
\begin{equation}
(1-2^{1-t})^{1+i+j}xy \leq \approx{x} \otimes \approx{y} \leq xy.
\label{eq:fl2}
\end{equation}

For each particular problem, we will choose $t$ as a function of $n$ and of the error factor $\eps$, $0 < \eps \leq 1$. For the problem of random generation of DAGs we have $c = 2$, and we take $t(n,\eps) = 1+ \log(3n^3/\eps)$; in the case of counting 0/1 Knapsack solutions $c = 1$ and $t(n,\eps) = 1 + \log (n/\eps)$, while for its extension on a DAG, $c = 1$ and $t(n,\eps) = 1 + \log (n\log(\frac{m}{n})/\eps)$.

\section{Random generation of DAGs}
\label{sec:uarGASs}

\begin{sloppypar}
In Sec.~\ref{sec:EXACT-DAG} we present the well-known decomposition of labeled DAGs by sources~\cite{R73}, and turn it into a deterministic random generation algorithm. In Sec.~\ref{sec:FPTAS-DAG}, we show how to approximate the numerical values of the counting recurrence, and argue that the resulting random generation algorithm is an FPTAS of lower complexity. 
\end{sloppypar}

%Note that further refinements of this algorithm can be made; for example, it is possible to employ the number of arcs as a third parameter in the counting recurrence of DAGs (cf.~\cite{B86}), and accordingly, obtain a random generation algorithm for DAGs with a given number of nodes and a given number of edges.

\subsection{Exact u.a.r.~generation of DAGs by sources}
\label{sec:EXACT-DAG}

For $1 \leq k \leq n$,
let $a({n,k})$ denote the number of labeled DAGs with $n$ vertices, out of which precisely $k$ are sources. 
Then $a(n):= \sum_{k=1}^n a({n,k})$ is the number of labeled DAGs on $n$ vertices.
In~\cite{R73} a simple decomposition of DAGs by sources delivers the following counting recurrence:

\begin{equation}
a({n,k}) = {n \choose k}\sum_{s=1}^{n-k}(2^k-1)^s 2^{k(n-k-s)} a({n-k,s}),
\label{eq-dags}
\end{equation}

\noindent where $a({n,n}) = 1$, for all $n \geq 1$. Indeed, there are ${n \choose k}$ ways to choose the $k$ sources, and by removing the sources we obtain a DAG with $n-k$ vertices and $s$ sources, $1 \leq s \leq n-k$. Each of these $s$ sources must have a non-empty in-neighborhood included in the set of $k$ removed vertices, while the other $n-k-s$ vertices can have arbitrary in-neighbors among these $k$ vertices.

In order to generate u.a.r.~a DAG having the set $V = \{0,\dots,n-1\}$ as vertex set, recurrence (\ref{eq-dags}) suggests the following recursive algorithm. Choose the number $k$ of its sources with probability $a(n,k)/\sum_{t=1}^na(n,t)$. Then, choose u.a.r.~the $k$ sources $\{v_1,\dots,v_k\}$, and call the recursive algorithm for $V \setminus \{v_1,\dots,v_k\}$. Finally, connect $\{v_1,\dots,v_k\}$ with the graph returned by the recursive call, as indicated by the proof of (\ref{eq-dags}); see Algorithm~\ref{generation-1}. 

\begin{algorithm}[t]
\footnotesize
\caption{\small {\sc randomGenerateDAG}($V,a(\cdot\,,\cdot)$)\protect\\Returns a random DAG on vertex set $V$, dubbed $D$, together with the set of its sources.\protect\\ The table of values $a(\cdot\,,\cdot)$ is either computed exactly or approximately, according to recurrence (\ref{eq-dags}).\label{generation-1}}

	$n := |V|$\;

	\textbf{if} $n = 0$ \textbf{then} \textbf{return} $((\emptyset,\emptyset),\emptyset)$\;
	
	choose $k \in \{1,\dots,n\}$ with probability $a(n,k)/\sum_{t=1}^{n}a(n,t)$\;
	choose u.a.r.~a $k$-subset $\{v_1,\dots,v_k\} \subseteq V$\;
	$(D,S) := \textsc{randomGenerateDAG}(V \setminus \{v_1,\dots,v_k\},a(\cdot\,,\cdot))$\;
	
	\BlankLine
	
	$X :=V(D) \setminus S$; $V(D) := V(D) \cup \{v_1,\dots,v_k\}$\;
	\ForEach{$s \in S$}{
		$N^-(s) :=$ a non-empty subset of $\{v_1,\dots,v_k\}$ chosen u.a.r.\;
	}

	\ForEach{$x \in X$}{
		$N^-(x) := N^-(x) \; \cup$ a subset of $\{v_1,\dots,v_k\}$ chosen u.a.r.\;
	}
	
	\textbf{return} $(D,\{v_1,\dots,v_k\})$.

\end{algorithm}

In order to choose a number $k \in \{1,\dots,n\}$ with probability $a(n,k)/\sum_{t=1}^{n}a(n,t)$, we can choose u.a.r.~$r \in \{1,\dots,\sum_{j=1}^{n}a({n,j})\}$, and then take $k$ as the smallest integer such that $r \leq \sum_{j=1}^{k}a({n,j})$. For every $i \in \{1,\dots,n\}$, we can store a patricia trie containing values $\sum_{t=1}^{j}a(i,t)$, for all $j \in \{1,\dots,i\}$; $r$ is found by a successor query in the patricia trie. 

The asymptotic behavior of $a(n)$ is $a(n) \sim n!2^{n \choose 2}/(Mp^n),$ where $M = 0.474$ and $p = 1.448...$~\cite{B86,B88}. Therefore, we need $O(n^2)$ bits to store each $a(n,k)$. In order to compute numbers $a(n,k)$, we assume to have access to pre-computed tables storing numerical values of binomial coefficients, and of all $(2^k-1)^s$; number $2^{k(n-k-s)}$ can be computed by setting one bit to 1. Each number $a(n,k)$ can be then computed with $O(n)$ additions and multiplications on $n^2$ bits. Therefore, computing the entire table $a(n,k)$ has bit complexity $O(n^5M(n))$.

For every $i \inset{1}{n}$, the $i$th patricia trie can be constructed with $O(n^3)$ bit operations, uses space $O(n^3 + n\log n)$ bits, and supports successor queries in time $O(n^2)$; these are standard considerations in data structures. Therefore, choosing $k$ takes time $O(n^2)$. The second part of the algorithm takes overall $O(n^2)$ time, since each of the $O(n^2)$ arcs of a DAG is introduced at most once. Therefore, we obtain Thm.~\ref{thm:exactDAGs}.

%\begin{quote}
%Mi pare che per lo stesso argomento di prima,~\cite{DAGS2012} conclude con una complessita' di $O(n^3)$. 
%\end{quote}

\subsection{An FPTAS for generating labeled DAGs u.a.r.}
\label{sec:FPTAS-DAG}

Let $\eps$, $0 < \eps \leq 1$, be fixed. Instead of using $n^2$ bits for storing each entry in the table $a(n,k)$, we use floating-point representations with $2\log n$ bits for the exponent and $t(n,\eps) = 1 + \log(3n^3/\eps)$ bits for the mantissa.

For each $1 \leq k \leq n$, we approximate $a(n,k)$ by $\approx{a}({n,k})$, recursively computed by floating-point additions and multiplications, as:

\begin{equation}
\approx{a}({n,k}) = \fl{{n \choose k}} \otimes \bigoplus_{s = 1}^{n-k}\left(\fl{(2^k - 1)^s}\otimes\fl{2^{k(n-k-s)}}\otimes \approx{a}({n-k,s})\right),
\label{eq:DAGsApprox}
\end{equation}

\noindent where $\approx{a}(k,k) = a(k,k) = 1$, for all $k \leq n$. In order to compute numbers $\approx{a}(n,k)$, we assume to have access to tables now storing floating-point approximations with $2\log n$ bits for the exponent and $1 + \log(3n^3/\eps)$ bits for the mantissa, with a precision as in (\ref{eq-1}), of binomial coefficients and of numbers $(2^k - 1)^s$. These floating-point numbers can be obtained from the tables storing their exact values, assumed available in the exact case, by trivially setting the exponent to be the length of the exact number, and by filling in its mantissa by taking the first $1 + \log(3n^3/\eps)$ bits. Number $2^{k(n-k-s)}$ can be represented exactly with the floating-point representation by setting the exponent to $k(n-k-s) + 1$, the first bit of the mantissa to 1, and the remaining bits to 0.

%\begin{algorithm}[h!]
%\small
%\caption{\small Approximating $a(n,k)$ by $\approx{a}(n,k)$\label{alg:count-1}}
%
%\SetKwBlock{approximateCountDAG}{{\sc approximateCountDAG}($n,k$)}{end}
%\approximateCountDAG{
%
%	\textbf{if} $k = n$ \textbf{then} \textbf{return} 1\;
%	
%	\BlankLine
%	
%	$\approx{a}(n,k,0) := 0$\;
%	\For{s = 1 \textbf{to} n-k}{
%		$a'(s) := \fl{\fl{(2^k - 1)^s}2^{k(n-k-s)}}$\;
%		$a''(s) := \fl{a'(s)\approx{a}({n-k,s})}$\;
%		$\approx{a}(n,k,s) := \fl{\approx{a}(n,k,s-1) + a'}$\;
%	}
%	$\approx{a}(n,k) := \fl{\approx{a}(n,k,n-k)\fl{{n \choose k}}}$\;
%	
%	\textbf{return} $\approx{a}(n,k)$.
%}
%\end{algorithm}

%\begin{algorithm}[h!]
%\footnotesize
%\caption{\small An FPTAS for generating a DAG with vertex set $V$; the procedure outputs a DAG $D$ and the set of its vertices\label{generation-3}}
%
%\SetKwBlock{randomGenerate}{{\sc approximateRandomGenerationDAG}($V$)}{end}
%
%\randomGenerate{
%
%	$n := |V|$\;
%
%	\textbf{if} $n = 0$ \textbf{then} \textbf{return} $((\emptyset,\emptyset),\emptyset)$\;
%	
%	choose $k \in \{1,\dots,n\}$ with probability $\approx{a}(n,k)/\sum_{t=1}^{n}\approx{a}(n,t)$\;
%	choose u.a.r.~a $k$-subset $\{v_1,\dots,v_k\} \subseteq V$\;
%	$(D,S) := \textsc{approximateRandomGenerationDAG}(V \setminus \{v_1,\dots,v_k\})$\;
%	
%	\BlankLine
%	
%	$X :=V(D) \setminus S$; $V(D) := V(D) \cup \{v_1,\dots,v_k\}$\;
%	\ForEach{$s \in S$}{
%		$N^-(s) :=$ a non-empty subset of $\{v_1,\dots,v_k\}$ chosen u.a.r.\;
%	}
%
%	\ForEach{$x \in X$}{
%		$N^-(x) := N^-(x) \; \cup$ a subset of $\{v_1,\dots,v_k\}$ chosen u.a.r.\;
%	}
%	
%	\textbf{return} $(D,\{v_1,\dots,v_k\})$.
%}
%
%\end{algorithm}

Each number $\approx{a}(n,k)$ can be computed with $O(n)$ floating-point additions and multiplications on $O(\log(n/\eps))$-bit numbers; thus, the entire table $\approx{a}(n,k)$ can be computed in time $O(n^3\log(n/\eps)M(\log (n/\eps)))$. 
%\begin{algorithm}[h!]
%\small
%\caption{\small Generating uniformly at random a DAG with $n$ vertices labeled by $\{1,\dots,n\}$\label{generation-4}}
%
%\SetKwBlock{randomGenerate}{{\sc approximateRandomGenerationDAG}($n$)}{end}
%\randomGenerate{
%	
%	choose uniformly at random $r \in \{1,\dots,\sum_{j=1}^{n}\approx{a}({n,j})\}$\;
%	let $s$ be the smallest integer such that $r \leq \sum_{j=1}^{s}\approx{a}({n,j})$\;
%	
%	\textbf{return} $\textsc{approximateRandomGenerationDAG}(n,s,\{1,\dots,n\})$.
%}
%\end{algorithm}

The following lemma characterizes the approximation quality of the numbers $\approx{a}(n,k)$.

\begin{lemma}
For any $n \geq 1$ and any $1 \leq k \leq n$, it holds that
\[\left(1-{2^{1-t(n,\eps)}}\right)^{3n^2}a({n,k}) \leq \approx{a}({n,k}) \leq a({n,k}).\]
\label{lemma:DAGs1}

\end{lemma}

\begin{proof}
We prove the first inequality; $\approx{a}({n,k}) \leq a({n,k})$ will follow analogously. We reason by induction on $n$, the claim being clear for $n = 1$. For any $1 \leq s \leq n-k$, it holds that
\[\left(1-2^{1-t(n,\eps)}\right)^{3+3(n-k)^2}(2^k-1)^s 2^{k(n-k-s)} a({n-k,s}) \leq \fl{(2^k - 1)^s}\otimes\fl{2^{k(n-k-s)}}\otimes \approx{a}({n-k,s}),\]

\begin{sloppypar}
\noindent since $\fl{2^{k(n-k-s)}} = 2^{k(n-k-s)}$, by (\ref{eq:fl2}) it holds that $\left(1-2^{1-t(n,\eps)}\right)^3 (2^k - 1)^s2^{k(n-k-s)}\leq \fl{(2^k - 1)^s}\otimes\fl{2^{k(n-k-s)}}$, and from the inductive hypothesis we have $\left(1-{2^{1-t(n,\eps)}}\right)^{3(n-k)^2}a({n-k,s}) \leq \approx{a}({n-k,s})$.
\end{sloppypar}
Since the sum goes over $s$ from 1 to $n-k$, we have to do $n-k-1$ floating-point additions, therefore, by (\ref{eq:fl1}),
\vspace{-.3cm}
\begin{equation*}
\begin{split}
\left(1-2^{1-t(n,\eps)}\right)^{3+3(n-k)^2 + n-k-1}\sum_{s=1}^{n-k}(2^k-1)^s 2^{k(n-k-s)} a({n-k,s}) \leq \\
\leq \bigoplus_{s=1}^{n-k}\fl{(2^k - 1)^s}\otimes\fl{2^{k(n-k-s)}}\otimes \approx{a}({n-k,s}).
\end{split}
\end{equation*}
\vspace{-.3cm}

We assumed that $(1-2^{1-t(n,\eps)}){n \choose k} \leq \fl{{n \choose k}}$, therefore this implies, by (\ref{eq:fl2}), that

\[\left(1-2^{1-t(n,\eps)}\right)^{3(n-k)^2 + n-k+4}a(n,k) \leq \approx{a}(n,k).\] 

Since $k \geq 1$ and $n \geq 2$, we have $3(n-k)^2 + n-k + 4 \leq 3(n-1)^2 + (n-1) + 4 = 3n^2 - 5n + 6 \leq 3n^2$, which proves the claim, because $1-2^{1-t(n,\eps)} < 1$.
\end{proof}

Lemma~\ref{lemma:DAGs1} immediately implies an FPTAS for counting labeled DAGs, as stated by Thm.~\ref{thm:counting-DAGs}. For completing the proof of Thm.~\ref{thm:counting-DAGs}, take $t(n,\eps) = 1 + \log(3n^2/\eps)$ and use relation~(\ref{eq:epsfptas}) below.

\bigskip

Notice that the table of numbers $\approx{a}(\cdot\,,\cdot)$ depends only on $n$ and $\eps$. We propose to run Algorithm~\ref{generation-1} on the table $\approx{a}(\cdot\,,\cdot)$. We use the same scheme as before for choosing $k$, which now takes time $O(\log (n/\eps))$. This is our FPTAS for approximate random generation.

\begin{proofSpec}
Let $D$ be a fixed DAG with $n$ vertices and assume that $F_1$, $f_1 := |F_1|$, is the set of sources of $D$, $F_2$, $f_2 := |F_2|$, is the set of sources of $D \setminus F_1$, and so on, until say $F_d$, with $f_d := |F_d|$. The probability of generating $D$ u.a.r., which is $1/a(n)$, can also be expressed, as a consequence of Algorithm~\ref{generation-1}, as
\begin{equation*}
\begin{split}
\frac{1}{a(n)} & = \frac{a({n,f_1})}{a({n,1}) + \dots + a({n,n})} \cdot \frac{a({n-f_1,f_2})}{a({n-f_1,1}) + \dots + a({n-f_1,n-f_1})} \cdot \cdots \\
& \cdot \frac{a({n-f_1-\cdots - f_{d-1},f_d})}{a({n-f_1-\cdots - f_{d-1},1}) + \dots + a({n-f_1-\cdots - f_{d-1},n-f_1-\cdots - f_{d-1}})} \\
& = \prod_{i=1}^d \frac{a({n - \sum_{j=1}^{i-1}f_j,f_i})}{\sum_{\ell = 1}^{n - \sum_{j=1}^{i-1}f_j} a({n - \sum_{j=1}^{i-1}f_j,\ell})}
\end{split}
\end{equation*}

The probability $\approx{p}(D)$ that {\sc randomGenerateDAG}$(V,\approx{a}(\cdot\,,\cdot)) = (D,F_1)$ is 

\[\prod_{i=1}^d \frac{\approx{a}({n - \sum_{j=1}^{i-1}f_j,f_i})}{\sum_{\ell = 1}^{n - \sum_{j=1}^{i-1}f_j} \approx{a}({n - \sum_{j=1}^{i-1}f_j,\ell}))}.\]

Therefore, by Lemma~\ref{lemma:DAGs1}, since $d \leq n$, $(1-2^{1-t(n,\eps)})^{3n^3} \leq a(n)\approx{p}(D) \leq (1-2^{1-t(n,\eps)})^{-3n^3}$ holds. If we choose $t(n,\eps) = 1 + \log(3n^3/\eps)$, it holds that
\vspace{-.2cm}
\[\left(1 - \frac{\eps}{3n^{3}}\right)^{3n^3} \leq a(n)\approx{p}(D) \leq \left(1 - \frac{\eps}{3n^{3}}\right)^{-3n^3}.\] 

By standard techniques, for all natural numbers $n \geq 1$ and all $0 < \eps \leq 1$, the following hold:
\begin{equation}
1-\eps \leq \left(1 - \frac{\eps}{n}\right)^{n}, \text{ and } \left(1 - \frac{\eps}{n}\right)^{-n} \leq 1+\eps.
\label{eq:epsfptas}
\end{equation}
\end{proofSpec}

\section{Counting 0/1 Knapsack solutions}

The classic pseudo-polynomial algorithm for counting 0/1 Knapsack solutions defines $s(i,c)$ as the number of Knapsack solutions that use a subset of the items $\{1,\dots,i\}$, of weight at most $c \inset{0}{C}$, and computes these values $s(i,c)$ by dynamic programming, using the recurrence
\begin{equation}
s(i,c) = s(i-1,c) + s(i-1,c-w_i).
\label{eq:k1}
\end{equation}
Indeed, we either use only a subset of items from $\{1,\dots,i-1\}$ whose weights sum up to $c$, or use item $i$ of weight $w_i$ and a subset of items from $\{1,\dots,i-1\}$ whose weights sum up to $c - w_i$. This DP algorithm executes $nC$ additions on $n$-bit numbers and its complexity is $O(C\,n^2)$.
When $C \leq n$, this is $O(n^3)$, whence $n \leq C$ will be assumed in the following.
We will assume, like in~\cite{Stefankovic:2012:DPA:2339941.2339944}, that additions and comparisons on numbers with $O(\log C)$ bits have unit cost,
which implies the same on $O(\log n)$-bit numbers. 

We also use relation (\ref{eq:k1}) to count, but our numbers, for any $0 < \eps \leq 1$, are approximate floating-point numbers with $\log n$ bits for the exponent, and $1 + \log (n/\eps)$ bits for the mantissa (we can assume for simplicity that a solution using all $n$ objects has cost greater than $C$, so that $s(i,c) < 2^n$ for all $i \inset{1}{n}$, $c \inset{0}{C}$). By the above assumption, we have that additions and comparisons of these floating-point numbers on $O(\log (n/\eps))$ bits take time $O(\lceil\log (n/\eps)/\log n\rceil) = O(\lceil\log(1/\eps)/\log n\rceil)$. %The intuition behind our choice is that the number of possible different entries in each row $i$ of the dynamic programming table is small, more precisely, is at most $2^{\log n + \log(2n/\eps)} = O(n^2/\eps)$.

%We store and compute these approximate numbers as follows. For each $i \inset{1}{n}$, we compute a list $\approx{s}(i)$ containing, for all $t \inset{0}{2^{\log n + \log(2n/\eps)}-1}$ such that $s(i,C) \geq t$, the pair $[c,t]$, where $c$ is the smallest capacity such that $s(i,c) \geq t$; the pairs will be ordered on the first component. 

\begin{algorithm}[t]
\footnotesize
\caption{\small {\sc ApproximatelyCountKnapsackSolutions}($w_1,\dots,w_n,C$) \protect\\
An FPTAS for counting 0/1 Knapsack solutions\label{alt:k2}}
\textbf{Notation:} $\approx{s}(i,c) := \max\{t \,:\, [c',t]\in list(i), \,c' \leq c \}$\;

\BlankLine

	\BlankLine
	insert the pair $[0,1]$ into $list(0)$\;
	\BlankLine
	
	\For{$i=1 \mathbf{\ to\ } n$}{
		construct the bimotonotic $list'(i)$ containing, for each $[c,t]$ in $list(i-1)$, the two pairs:\\
		\smallskip
		\hspace{.5cm} $\bullet$~~$[c,t \oplus \approx{s}(i-1,c-w_i)]$\;
		\hspace{.5cm} $\bullet$~~$[c + w_i,\approx{s}(i-1,c+w_i) \oplus t]$\;
		\smallskip
		obtain $list(i)$ by scanning $list'(i)$ and dropping a pair if the previous one has the same second component\;
		}	
	\textbf{return} $\approx{s}(n,C)$.

\end{algorithm}

For every $i \inset{0}{n}$ we keep a list, $list(i)$, whose entries are pairs of the form $[c,t]$, where $c$ is a capacity in $\{0,\dots,C\}$ and $t$ is an approximate floating-point number of solutions. We will refer to the set of first components of the pairs in $list(i)$ as the \emph{capacities in $list(i)$}.

Having $list(i)$, for every $c \inset{0}{C}$ we define $\approx{s}(i,c) := \max\{t \,:\, [c',t]\in list(i), \,c' \leq c \}$,
where the maximum of an empty set is taken to be $0$.

The first list, $list(0)$, consists of the single pair $[0,1]$. After this initialization,
while computing $list(i)$ from $list(i-1)$,
we maintain the following two invariant properties:

\begin{itemize}
\item[\textsf{(I$_{1}$)}] $list(i)$ is strictly increasing on both components;
%\item the length of $list(i)$ is $O(n^2/\eps)$;
\item[\textsf{(I$_{2}$)}] $\left(1-\eps/n\right)^i s(i,c) \leq \approx{s}(i,c) \leq s(i,c)$, for every $c \inset{0}{C}$.
\end{itemize}

Note that Property~\textsf{(I$_{1}$)} implies that the length of $list(i)$ is at most the total number of floating-point numbers that can be represented with $\log n + \log (n/\eps) + 1$ bits, that is $O(n^2/\eps)$.\\

We obtain $list(i)$ by first building the bimonotonic list $list'(i)$ which, for every capacity $c$ in $list(i-1)$, contains the following two pairs:
\begin{equation}
[c,\approx{s}(i-1,c) \oplus \approx{s}(i-1,c-w_i)] \text{ and } [c + w_i,\approx{s}(i-1,c+w_i) \oplus \approx{s}(i-1,c)].
\label{eq:kk1}
\end{equation}

It may turn out that $list'(i)$ contains distinct pairs having the same second component. Therefore,
in order to assure Property~\textsf{(I$_{1}$)}, we obtain $list(i)$ by pruning away from $list'(i)$ those pairs $[c_2,t]$ when another pair $[c_1,t]$ with $c_1 < c_2$ is present. We summarize this procedure as Algorithm~\ref{alt:k2}. Lemma~\ref{lemma:list'} below shows that we can efficiently construct $list'(i)$; the idea of the proof is to do two linear scans of $list(i)$, each with two pointers, and it is in Appendix~\ref{appendix:proof-list'}.

\begin{lemma}
We can compute $list'(i)$ and $list(i)$ from $list(i-1)$ in time $O(n^2\eps^{-1}\lceil\log(1/\eps)/\log n\rceil)$.
\label{lemma:list'}
\end{lemma}

\begin{lemma}
Property~\textsf{(I$_{2}$)} holds for $list(i)$, that is, for every $i \inset{0}{n}$ and every $c \inset{0}{C}$, $\left(1-\eps/n\right)^i s(i,c) \leq \approx{s}(i,c) \leq s(i,c)$ holds.
\label{lemma:list(i)-approx}
\end{lemma}

\vspace{-.1cm}

\begin{proof}
The claim is clear for $i = 0$. For an arbitrary capacity $c \inset{0}{C}$, let $[c_1,t_1]$ in $list(i)$ be such that $\approx{s}(i,c) = t_1$. From the definition of $\approx{s}$, we get $\approx{s}(i,c) = \approx{s}(i,c_1)$; from the fact that the pairs in $list(i)$ are of the form (\ref{eq:kk1}), we have 
\begin{equation}
\approx{s}(i,c) = \approx{s}(i,c_1) = \approx{s}(i-1,c_1) \oplus \approx{s}(i-1,c_1 - w_{i}).
\label{eq:kk2}
\end{equation}
Since the capacities in $list(i-1)$ are a subset of the capacities in $list'(i)$, and the fact that we have pruned the pairs in $list'(i)$ by keeping the smallest capacity for every approximate number of solutions corresponding to that capacity, it holds that $\approx{s}(i-1,c_1) = \approx{s}(i-1,c)$.
Moreover, observe that there is no capacity $c_2$ in $list(i-1)$ such that $c_1-w_{i} < c_2 < c - w_{i}$. Indeed, for assuming the contrary, $c_2 + w_{i}$ would be a capacity in $list'(i)$, by~(\ref{eq:kk1}). Since we have chosen $c_1$ as the largest capacity in $list(i)$ smaller than $c$, and $c_1 < c_2 + w_{i} < c$ holds, this implies that $c_2 + w_{i}$ was pruned when passing from $list'(i)$ to $list(i)$; thus, the two pairs of $list'(i)$ having $c_1$ and $c_2 + w_{i}$ as first components have equal second components. By (\ref{eq:kk2}) and the bimonotonicity of $list(i-1)$, this entails that also the two pairs of $list(i-1)$ having $c_1-w_{i}$ and $c_2$ as first components must have equal second components. This contradicts the fact that $list(i-1)$ satisfies Property \textsf{(I$_{1}$)}.

Therefore, it also holds that $\approx{s}(i-1,c_1 - w_{i}) = \approx{s}(i-1,c - w_{i})$.
Plugging these two relations into~(\ref{eq:kk2}) we obtain 
\begin{equation}
\approx{s}(i,c) = \approx{s}(i-1,c) \oplus \approx{s}(i-1,c - w_{i}).
\label{eq:kk3}
\end{equation}
From~(\ref{eq:k1}), the fact that Property~\textsf{(I$_{2}$)} holds for $list(i-1)$, and from~(\ref{eq:fl1}), we get that $(1-\eps/n)^{i}s(i,c) \leq \approx{s}(i,c) \leq s(i,c)$, which shows that Property~\textsf{(I$_{2}$)} holds also for $list(i)$.
\end{proof}

From Lemma~\ref{lemma:list(i)-approx}, the fact that Property~\textsf{(I$_{2}$)} holds, and (\ref{eq:epsfptas}), we finally obtain Thm.~\ref{thm:knapsack-main}. Since $n^2\eps^{-2} = \Omega(n^3\eps^{-1}\lceil \log(1/\eps)/\log n\rceil)$ when $\eps^{-1} = \Omega(n)$, our deterministic FPTAS also runs faster than the Monte Carlo FPRASes in~\cite{DBLP:conf/stoc/Dyer03}, which currently held the record on the whole range, as soon as $\eps = o(1/n)$.

\bigskip

Onwards, we briefly sketch the details on applying this method to the Knapsack problem on a DAG (the full explanation is available in Appendix~\ref{appendix:kDAG}.). We can assume that all vertices of the DAG $D$ (with $n$ vertices and $m$ arcs) are reachable from $s$, and all vertices reach $t$. For simplicity, we transform $D$ into an equivalent DAG $D'$ in which every vertex has at most two in-coming arcs, and $D'$ has $O(m)$ vertices and arcs, and the maximum path length (i.e., number of arcs in the path) is $O(n\log(\frac{m}{n}))$. Say that $D'$ has $n'$ vertices and let $s = v_1,v_2,\dots,v_{n'} = t$ be a topological ordering of them. We now denote by $s(i,c)$ the number of paths that end in $v_i$ and their total weight is at most $c \inset{0}{C}$. If for every node $v_i$, its in-degree is $d(i)$, its in-neighborhood is $\{v_{i_1},v_{i_{d(i)}}\}$, and the weights of the arcs entering $v_i$ are $w_{i_1},w_{i_{d(i)}}$, respectively, relation (\ref{eq:k1}) generalizes to:

\begin{equation}
s(i,c) = 
\begin{cases}
s(i_1,c-w_{i_1}), & \text{if $d(i) = 1$,} \\
s(i_1,c-w_{i_1}) + s(i_2,c-w_{i_2}), & \text{if $d(i) = 2$.}
\end{cases}
\label{eqn:kDAG2-m}
\end{equation}

The solution is obtained as $s(n',C)$. As before, we use (\ref{eqn:kDAG2-m}) to count, keeping at each step approximate floating-point numbers. These numbers still have $\log n$ bits for the exponent, but, since the maximum path length in $D'$ is $O(n\log(\frac{m}{n}))$, the length of their mantissa will be $1 + \log(n\log(\frac{m}{n})/\eps)$ bits. Additions and comparisons of these floating-point numbers still take the same time as before, namely $O(\lceil\log(1/\eps)/\log n\rceil)$.

As before, for every $i \inset{1}{n'}$, we keep a list, $list(i)$, of pairs [capacity, approximate number of solutions], now of length at most $O(n^2\log(\frac{m}{n})\eps^{-1})$. Analogously, $list(1)$ consists of the single pair $[0,1]$, and while computing $list(i)$ from lists $list(i_1)$, or $list(i_1)$ and $list(i_2)$ (doable now in time $O(n^2\log(\frac{m}{n})\eps^{-1}\lceil\log(1/\eps)/\log n\rceil)$), we maintain the following two invariants, where $\ell(i)$ denotes the length of the longest path from $s$ to $v_i$:

\begin{itemize}
\item[\textsf{(I$_{1}$)}] $list(i)$ is strictly increasing on both components;
\item[\textsf{(I$_{2}$)}] $\left(1-\eps/(n\log(\frac{m}{n}))\right)^{\ell(i)} s(i,c) \leq \approx{s}(i,c) \leq s(i,c)$, for every $c \inset{0}{C}$.
\end{itemize}

From these considerations, Thm.~\ref{thm:kDAG-main} immediately follows.

%\section{Conclusions}
%
%In the same way that the random generation problem for DAGs was first tackled by a Markov Chain algorithm, with no proof of rapid mixing, the general point we want to make with these examples is that 

\newpage

\section*{Acknowledgements}

This work was partially supported by the Academy of Finland under grant 250345 (CoECGR), and by the European Science Foundation, activity ``Games for Design and Verification''. We thank Djamal Belazzougui and Daniel Valenzuela for discussions on data structures for large numbers, and Stephan Wagner for remarks on decomposable structures.

\bibliographystyle{siam}
\bibliography{generation}

\newpage
\appendix

\section{Random generation of other DAG subclasses}
\label{appendix:DAGs}

\subsection{Essential DAGs}

Essential DAGs are used to represent the structure of Bayesian networks~\cite{Gillispie01enumeratingmarkov,DBLP:journals/jmlr/Pena07}. They were counted in~\cite{DBLP:journals/dm/Steinsky03} by inclusion-exclusion, and their asymptotic behavior was studied in~\cite{DBLP:journals/gc/Steinsky04}. We give a new counting recurrence for essDAGs, which leads to the first algorithm for generating u.a.r.~a labeled essDAG with $n$ vertices; this is useful for learning the structure of a Bayesian network from data~\cite{Gillispie01enumeratingmarkov,DBLP:journals/jmlr/Pena07}. This can be turned into an FPTAS, with the same complexity and approximation bounds as in the case of DAGs.

Essential DAGs (essDAGs) are those DAGs with the property that for every edge $(u,v)$, the set of in-neighbors of $u$ is different from the set of in-neighbors of $v$, minus vertex $u$; that is, for every $(u,v) \in E(D)$ it holds that $N^-(u) \neq N^-(v) \setminus \{u\}$.

Define the \emph{depth} of a vertex $x$ in a DAG $D$ as the length of any longest directed path from a source of $D$ to $x$. Note that a vertex of maximum depth in $D$ must be a sink of $D$ (but the converse does not hold). Let us denote by $d(n,k)$ the number of labeled essDAGs with $n$ vertices, and in which there are $k$ vertices of maximum depth. 

%\begin{algorithm}[t]
%\small
%\caption{\small Generating uniformly at random an essDAG with $n$ vertices labeled by the elements of $V$, $|V| = n$, out of which $k$ are vertices of maximum depth\label{generation-5}}
%
%\SetKwBlock{randomGenerate}{{\sc randomGenerationEssDAG}($n,k,V$)}{end}
%\randomGenerate{
%
%	\textbf{if} $k = n$ \textbf{then} \textbf{return} $(V,\emptyset)$\;
%	
%	choose uniformly at random a $k$-subset $Z \subseteq V$\;
%	choose $s \in \{1,\dots,n-k\}$ with probability $d(n-k,s)/\sum_{t=1}^{n-k}d(n-k,t)$\;
%	$D := \textsc{randomGenerationEssDAG}(n-k,s,V \setminus Z)$\;
%
%	\BlankLine
%		
%	$X := $ the set vertices of maximum depth of $D$\;
%	$Y :=V(D) \setminus X$\;
%	$V(D) := V(D) \cup Z$\;
%	\ForEach{$z \in Z$}{
%		choose $b \in \{0,1\}$ at random such that
%		\[b = \begin{cases}
%		0, & \text{with probability $s(2^{n-k-s} - 1) / \left(s(2^{n-k-s} - 1) + (2^s - s -1)2^{n-k-s}\right)$\;}\\
%		1, & \text{with probability $(2^s - s -1)2^{n-k-s} / \left(s(2^{n-k-s} - 1) + (2^s - s -1)2^{n-k-s}\right)$\;}
%		\end{cases}\]
%
%		\eIf{$b = 0$}{
%			choose uniformly at random $x \in X$\;
%			choose uniformly at random a subset $W$ of $Y$, different from $N^-(x)$\;	
%			$N^-(z) := W \cup \{x\}$\;
%		}{
%			choose uniformly at random a subset $W_1$ of $X$ with at least 2 elements\;
%			choose uniformly at random a subset $W_2$ of $Y$\;
%			$N^-(z) := W_1 \cup W_2$\;
%		}
%	}
%	
%	\textbf{return} $D$.
%}
%\end{algorithm} 

\begin{lemma}
For any $n \geq 1$ and any $1 \leq k \leq n$, the following recurrence relation holds, where $d(n,n) = 1$, for all $n \geq 1$, 

\[d(n,k) = {n \choose k}\sum_{s = 1}^{n-k} d(n-k,s)\left(s(2^{n-k-s} - 1) + (2^s - s -1)2^{n-k-s}\right)^k.\]

\label{lemma:essDAGs}
\end{lemma}

\begin{proof}
There are ${n \choose k}$ ways to choose the $k$ vertices of maximum depth, and by removing them we obtain an essDAG with $n-k$ vertices and $s$ vertices of maximum depth, for some $1 \leq s \leq n - k$. Each vertex $x$ of maximum depth must have an in-neighbor among these $s$ vertices. We distinguish two cases. 

First, $x$ has precisely one in-neighbor $y$ among these $s$ vertices, in which case any subset of the remaining $n-k-s$ vertices, except for the in-neighborhood of $y$, can act as in-neighborhood of $x$, restricted to these $n-k-s$ vertices; thus, there are $s(2^{n-k-s} - 1)$ ways of choosing the in-neighborhood of $x$. Second, $x$ has at least two neighbors among the $s$ vertices, in which case any subset of the remaining $n-k-s$ vertices can act as in-neighborhood of $x$, restricted to these $n-k-s$ vertices; this is true since no vertex among the $n-k$ vertices can have an in-neighbor among the $s$ vertices of maximum depth; thus, there are $(2^s - s -1)2^{n-k-s}$ ways of choosing the in-neighborhood of $x$.
\end{proof}

Thanks to the proof of Lemma~\ref{lemma:essDAGs}, in order to generate u.a.r.~an essDAG having $V = \{0,\dots,n-1\}$ as vertex set, proceed recursively as in the case of DAGs. Choose its number $k$ of vertices of maximum depth proportional to $d(n,k)/\sum_{t=1}^{n}d(n,k)$. Then, choose u.a.r.~the $k$ vertices of maximum depth $\{v_1,\dots,v_k\}$, and call the recursive algorithm for $V \setminus \{v_1,\dots,v_k\}$, which returns an essDAG $D$. Finally, for each $i \inset{1}{k}$, choose $b \in \{0,1\}$ at random such that
\begin{equation}		
b = \begin{cases}
		0, & \text{with probability $s(2^{n-k-s} - 1) / \left(s(2^{n-k-s} - 1) + (2^s - s -1)2^{n-k-s}\right)$;}\\
		1, & \text{with complementary probability.
		%$(2^s - s -1)2^{n-k-s} / \left(s(2^{n-k-s} - 1) + (2^s - s -1)2^{n-k-s}\right)$
		}
		\end{cases}
\label{eq:essDAGs1}
\end{equation}
If $b=0$, then choose $x$ u.a.r.~among the vertices of maximum depth of $D$, and choose u.a.r.~a subset $W$, different from $N^-(x)$, of the other vertices of $D$, and set $N^-(v_i) = \{x\} \cup W$. The set $W$ can be chosen in time $O(n)$, as we can encode $N^-(x)$ as an $n$-bit number $f$ having a 1 on bit $i$ iff vertex $i$ belongs to $N^-(x)$, for each $i \inset{0}{n-1}$; we then generate u.a.r.~a number $w \inset{0}{2^n-2}$. If $f \leq w$, we set $w = w + 1$. The set $W$ is such that vertex $i \in W$ iff the $i$th bit of $w$ is 1, for all $i \inset{0}{n-1}$.

Otherwise, if $b=1$, choose u.a.r.~a subset $W_1$ of at least two elements of the vertices of maximum depth of $D$, and choose u.a.r.~a subset $W_2$ of the other vertices of $D$, and set $N^-(v_i) = W_1 \cup W_2$. Using the same implementation details and as argued in the case of DAGs, the following theorem holds.

\begin{theorem}
A labeled essDAG with $n$ vertices can be generated u.a.r.~in time $O(n^3)$, provided a table of size $O(n^4)$ bits, computable in time $O(n^5M(n))$, is available.
\end{theorem}

We state the bounds for the FPTAS for essDAGs together with the one for extensional DAGs at the end of Sec.~\ref{appendix:extDAGs} below.

\subsection{Extensional DAGs}
\label{appendix:extDAGs}

Extensional DAGs (extDAGs) are used in set theory and in some programming languages to model hereditarily finite sets (see, e.g.,~\cite{MT10}). The first recurrence for counting them precedes the one for DAGs~\cite{P62}; their asymptotic behavior was studied recently~\cite{DBLP:conf/analco/Wagner12,Wagner:2012fk}. We gave a random generation algorithm based on ranking and unranking functions~\cite{DBLP:journals/ipl/RizziT13}, which can generate a labeled extDAGs with $n$ vertices in time $O(n^3)$, once an auxiliary table of size $O(n^4)$ bits, computable in time $O(n^5M(n^2))$, is available. We now argue that the decomposition of extDAGs by sources from~\cite{PolicritiT201X-IPL} leads to a similar recursive random generation algorithm; using floating-point numbers this also can be turned into an FPTAS, with the same complexity and approximation bounds as in the case of DAGs.

Extensional DAGs are the ones with the property that their vertices have pairwise different sets of out-neighbors, that is, for all distinct vertices $x,y$, $N^+(x) \neq N^+(y)$ holds. 
%In what follows we give an algorithm which generates u.a.r.~a labeled extDAG with $n$ vertices. Note that every extDAG is a rigid directed graph, in the sense that every possible automorphism is the identity morphism mapping every vertex onto itself (see, e.g.,~\cite{PolicritiT201X-IPL}). Therefore, this will also be an algorithm for generating unlabeled extDAGs (as customary, we regard two unlabeled extDAGs are \emph{identical} if they are isomorphic).
Let $e(n,k)$ denote the number of labeled extDAGs with $n$ vertices out of which $k$ are sources. Adapting a recurrence from~\cite{PolicritiT201X-IPL} to labeled extDAGs, we obtain that $e(n,k)$ satisfies the following recurrence relation:
\begin{equation}
e(n,k) = n\left((2^{n-k} - (n-1)) e(n-1,k-1) + \sum_{t = 0}^{n - k - 1} {k + t \choose t+1}2^{n-1 - (k + t)} e(n-1,k + t)\right),
\end{equation}
\noindent where $e(1,1) = 1$, and we interpret $e(n,0)$ as 0, for all $n \ge 2$. Indeed, an extDAG on $n \geq 2$ vertices and $k$ sources is obtained by the addition of a source in two ways. First, a source can be added to an extDAG on $n-1$ vertices and $k-1$ sources. As this source can have as out-neighbors only vertices which are not sources, and it must have its set of out-neighbors different from that of any of the other $n-1$ vertices, there are $(2^{n-k} - (n-1)) e(n-1,k-1)$ ways to add it. Second, a new source can be added to an extDAG on $n-1$ vertices and $k+t$ sources, for $t \inset{0}{n - k - 1}$, by connecting the new source with exactly $t + 1$ existing sources. This new source can have arbitrary arcs toward the remaining $n - 1 - (k + t)$ vertices since its set of out-neighbors will be different from any other of the $n - 1$ vertices. Hence, there are ${s + k \choose k + 1} 2^{n - 1 - (s + k)}e(n-1,s+k)$ ways to add it. 

In order to generate an extDAG with vertex set $V = \{0,\dots,n-1\}$, proceed recursively as before. Choose u.a.r.~a vertex $x$ to be a source, and call the recursive algorithm for $V \setminus \{x\}$, which returns an extDAG $D$. Choose $b \in \{-1,0,1,\dots,n-k-1\}$ at random such that
\begin{equation}		
b = \begin{cases}
		-1, & \text{with probability $n (2^{n-k} - (n-1)) e(n-1,k-1) / e(n,k)$;}\\
		t \in \{0,\dots,n-k-1\}, & \text{with probability $n{k + t \choose t+1}2^{n-1 - (k + t)} e(n-1,k + t) / e(n,k)$.
		}
		\end{cases}
\label{eq:extDAGs1}
\end{equation}
If $b=-1$, then choose u.a.r.~a subset $X$ of vertices of $D$ which are not sources, such that $X$ is different from the out-neighborhood of any other vertex, and set $N^+(x) = X$. This can also be done in time $O(n)$ by generalizing the idea exposed for essDAGs. Keep a patricia trie containing, for all $v \in V(D)$, their binary encodings $f_v$ such that the $i$th bit of $f_v$ is 1 iff vertex $i$ belong to $N^+(v)$. Then generate u.a.r.~$w \inset{0}{2^n-n-k-1}$ and look up in the patricia trie how many binary encodings lexicographically precede, or are equal to $w$, say $p$. Set $w := w + p$, decode the string of bits $w$ to obtain the set $X$. Finally, update the patricia trie by inserting $w$.

Otherwise, if $b \geq 0$, choose u.a.r.~a $(b+1)$-subset $Z$ of the sources of $D$, choose u.a.r.~a subset $X$ of the other vertices of $D$, and set $N^+(x) = X \cup Z$. Using the same implementation details and as argued in the case of DAGs, the following theorem holds.

\begin{theorem}
A labeled extDAG with $n$ vertices can be generated u.a.r.~in time $O(n^3)$, provided a table of size $O(n^4)$ bits, computable in time $O(n^5M(n))$, is available.
\end{theorem}

Just as in the case of DAGs, for $0 < \eps \leq 1$, instead of using $n^2$ bits for storing each entry in the tables $d(n,k)$ or $e(n,k)$, we store floating-point approximations with $2\log n$ bits for the exponent and $O(\log(n/\eps))$ bits for the mantissa. We can compute the approximated tables recursively, as done in (\ref{eq:DAGsApprox}). We also have to use approximate floating-point numbers the coefficients involved when choosing $b$ in (\ref{eq:essDAGs1}) and in (\ref{eq:extDAGs1}). Just as in the case of DAGs, the following holds, where $d(n)$ and $e(n)$ denote the number of labeled essDAGs, and labeled extDAGs, respectively, with $n$ vertices.

\begin{theorem}
For any $n \geq 1$, and for every $0 < \eps \leq 1$, a labeled essDAG, or a labeled extDAG, $D$ with $n$ vertices can be generated at random with probability $\approx{p}(D)$ such that 
\[1-\eps \leq \approx{p}(D)d(n) \leq 1+\eps, \text{ or } 1-\eps \leq \approx{p}(D)e(n) \leq 1+\eps, \text{ respectively,}\]
in time $O(n^2 + n\log (n/\eps))$, provided a table of size $O(n^2\log(n/\eps))$ bits, computable in time $O(n^3\log(n/\eps)M(\log (n/\eps)))$, is available.
\end{theorem}

\section{Additional proof for counting 0/1 Knapsack solutions}
\label{appendix:proof-list'}

\begin{proofSpecLemma}
At a generic step $i \inset{1}{n}$, we compute $list'(i)$ as follows. We construct two auxiliary lists of pairs $back(i)$ and $forw(i)$. For every capacity $c$ in $list(i-1)$, the list $back(i)$ will contain the pairs $[c,\approx{s}(i-1,c) \oplus \approx{s}(i-1,c-w_{i})]$, and the list $forw(i)$ will contain the pairs $[c + w_{i},\approx{s}(i-1,c+w_{i}) \oplus \approx{s}(i-1,c)]$. List $list'(i)$ is now obtained by merging in a unique sorted list the lists $back(i)$ and $forw(i)$.

In order to compute $forw(i)$, proceed as follows (the computation of $back(i)$ is entirely analogous). Keep two pointers $left$ and $right$ in $list(i-1)$. Pointer $left$ is initially set to the first pair in $list(i-1)$, say $[c,t]$. Pointer $right$ is also set to the first pair in $list(i-1)$, but starts scanning $list(i-1)$ until reaching a pair $[c_1,t_1]$, such that $c_1 +w_i\geq c$ and either $[c_1,t_1]$ is the last pair in $list(i-1)$, or $[c_1,t_1]$ is immediately followed by a pair $[c_2,t_2]$ with the property $c + w_i < c_2$. Append the pair $[c+w_i,t_1 \oplus t]$ at the end of $forw(i)$, and advance pointer $left$ to the next pair in $list(i-1)$; repeat the above procedure, by advancing pointer $right$ to the corresponding pair, and inserting a new resulting pair in $forw(i)$. This is repeated until pointer $left$ reaches the end of $list(i-1)$.

Observe that list $forw(i)$ is bimonotonic, by the fact that Property~\textsf{(I$_{1}$)} holds for $list(i-1)$. By analogy, this is true also for $back(i)$. Therefore, we can merge them and call $list'(i)$ the resulting list. In order to prune the bimonotonic list $list'(i)$ to obtain $list(i)$, we do a linear scan with two pointers, dropping a pair if the previous one has the same second component. Thus Property~\textsf{(I$_{1}$)} holds for $list(i)$. 

Since we assume that additions and comparisons on $O(\log C)$-bit numbers take unit time, that floating-point additions and comparisons take $O(\lceil\log(1/\eps)/\log n\rceil)$ time, and the length of $list(i-1)$ is $O(n^2/\eps)$, the construction of $list(i)$ takes time $O(n^2\eps^{-1}\lceil\log(1/\eps)/\log n\rceil)$.
\end{proofSpecLemma}

\section{Counting Knapsack solutions on a DAG}
\label{appendix:kDAG}

Without loss of generality, we can assume that all vertices of the DAG $D$ (with $n$ vertices and $m$ arcs) are reachable from $s$, and all vertices reach $t$; we also assume that the vertices are labeled in topological order $v_1,\dots,v_n$, such that $s = v_1$ and $t = v_n$. The dynamic programming for the 0/1 Knapsack problem can trivially be extended to a DAG. We denote by $s(i,c)$ the number of $s,v_i$-paths (clearly, these use a subset of the vertices $\{v_1,\dots,v_{i-1}\}$) and of total weight at most $c \inset{0}{C}$. If for every node $v_i$, its in-degree is $d(i)$, its in-neighborhood is $\{v_{i_1},\dots,v_{i_{d(i)}}\}$, and the weights of the arcs from each of these $d(i)$ in-neighbors are $w_{i_1},\dots,w_{i_{d(i)}}$, respectively, relation (\ref{eq:k1}) generalizes to:

\begin{equation}
s(i,c) = \sum_{j=1}^{d(i)}s(i_j,c-w_{i_j}),
\label{eqn:kDAG1}
\end{equation}
where we take $s(1,c) = 1$, for every $c \inset{0}{C}$, and $s(i,c) = 0$ for every $c < 0$ and every $i \inset{1}{n}$. The solution is obtained as $s(n,C)$. Since the number of all $s,v_i$-paths in the DAG is $O(2^i)$, for every $i \inset{1}{n}$, this DP executes $mC$ additions on $n$-bit numbers, and its complexity is $O(Cmn)$. Thus we can assume that $n \leq C$, and as before, that additions on $O(\log C)$-bit, and thus $O(\log n)$-bit numbers, have unit cost.

As in our solution for the 0/1 Knapsack problem, we use the DP to count, keeping at each step approximate floating-point numbers. These numbers still have $\log n$ bits for the exponent, but the length of their mantissa will be chosen to reflect the number of successive floating-point additions necessary to obtain $s(n,C)$. We can organize this computation in sequences of $O(n\log \frac{m}{n})$ repeated additions, and thus we can take the mantissa to be $1 + \log(n\log(\frac{m}{n})/\eps)$ bits long. Accordingly, additions and comparisons of these floating-point numbers still take the same time as before, namely $O(\lceil\log(1/\eps)/\log n\rceil)$.

For clarity, we explain how we organize the computation by transforming the input DAG $D$ into a DAG $D'$ in which every vertex has at most two in-neighbors. For every node $v_i$ of $D$, if $d(i) > 2$, we construct a complete binary tree on top of the in-neighbors $v_{i_1},\dots,v_{i_{d(i)}}$ of $v_i$, where $v_i$ is its root; this tree has $O(d(i))$ vertices and edges, and depth $\log(d(i))$. The vertices and edges of this tree are added to $D$, the arcs from $v_{i_1},\dots,v_{i_{d(i)}}$ to $v_i$ are removed, and all edges of the tree are directed towards $v_i$. Moreover, the weights of the new arcs out-going from $v_{i_1},\dots,v_{i_{d(i)}}$ are set to be the weights of their former arcs towards $v_i$; all other new arcs have weight 0. After transforming the in-neighborhood of all vertices of $D$, the original solutions are in one-to-one correspondence with the solutions of the transformed DAG $D'$. 

The DAG $D'$ has $O(m)$ vertices and $O(m)$ edges. Moreover, since $\sum_{i=1}^n d(i) = m$, the length of a path in $D'$ is at most $\max \sum_{i=1}^n \log d_i = \max \log \prod_{i=1}^n d_i$, where the maximum goes over all partitions of $m$ into $n$ integers $d_1,\dots,d_n$; the maximum is obtained when all factors of the product are $\Theta(m/n)$. Thus, the length of the longest path in $D'$ is $O(n\log(\frac{m}{n}))$.

We denote by $n'$ the number of vertices of $D'$, and we assume that $v_1,\dots,v_{n'}$ is a topological order on $D'$ (so that $s = v_1$ and $t = v_{n'}$). Using the same notation as above, relation (\ref{eqn:kDAG1}) simplifies to 
\begin{equation}
s(i,c) = 
\begin{cases}
s(i_1,c-w_{i_1}), & \text{if $d(i) = 1$,} \\
s(i_1,c-w_{i_1}) + s(i_2,c-w_{i_2}), & \text{if $d(i) = 2$.}
\end{cases}
\label{eqn:kDAG2}
\end{equation}

As in the case of the 0/1 Knapsack problem, for every $i \inset{1}{n'}$, we keep a list $list(i)$ of pairs [capacity, approximate number of solutions], and use the notation $\approx{s}(i,c)$ with the same meaning. Analogously, $list(1)$ consists of the single pair $[0,1]$, and while computing $list(i)$ from lists $list(i_1)$, or from $list(i_1)$ and $list(i_2)$, we maintain the following two invariants, where $\ell(i)$ denotes the length of the longest path from $s$ to $v_i$:

\begin{itemize}
\item[\textsf{(I$_{1}$)}] $list(i)$ is strictly increasing on both components;
\item[\textsf{(I$_{2}$)}] $\left(1-\eps/(n\log(\frac{m}{n}))\right)^{\ell(i)} s(i,c) \leq \approx{s}(i,c) \leq s(i,c)$, for every $c \inset{0}{C}$.
\end{itemize}

Property~\textsf{(I$_{1}$)} implies now that the length of $list(i)$ is $O(n^2\log(\frac{m}{n})\eps^{-1})$. If $d(i) = 1$, then we build $list(i)$ by scanning $list(i_1)$ and for every pair $[c_1,t_1]$, we insert the pair $[c_1+w_{i_1},t_1]$ in $list(i)$. It is obvious that the resulting list satisfies Properties \textsf{(I$_{1}$)} and \textsf{(I$_{2}$)}. 

Therefore, we consider onwards the case $d(i) = 2$. Analogously to (\ref{eq:kk1}), we first build $list'(i)$ that for every capacity $c_1$ in $list(i_1)$, contains the pair
\begin{equation}
[c_1+w_{i_1},\approx{s}(i_1,c_1) \oplus \approx{s}(i_2,c_1 + w_{i_1} - w_{i_2})],
\label{eqn:kDAG3}
\end{equation}
and for every capacity $c_2$ in $list(i_2)$, contains the pair 
\begin{equation}
[c_2+w_{i_2},\approx{s}(i_1,c_2+w_{i_2}-w_{i_1}) \oplus \approx{s}(i_2,c_2)].
\label{eqn:kDAG4}
\end{equation}
We obtain $list(i)$ by scanning $list'(i)$ and dropping a pair if the previous one has the same second component. We next prove an analog of Lemma~\ref{lemma:list(i)-approx}.

\begin{lemma}
Property~\textsf{(I$_{2}$)} holds for $list(i)$, that is, for every $i \inset{1}{n}$ and every $c \inset{0}{C}$, $\left(1-\eps/(n\log(\frac{m}{n}))\right)^{\ell(i)} s(i,c) \leq \approx{s}(i,c) \leq s(i,c)$ holds.
\label{lemma:list'-kDAG-0}
\end{lemma}

\begin{proof}
The claim is clear for $i = 1$. For an arbitrary capacity $c \inset{0}{C}$, let $[c_0,t_0]$ in $list(i)$ be such that $\approx{s}(i,c) = t_0$. Therefore, from the definition of $\approx{s}$, we get $\approx{s}(i,c) = \approx{s}(i,c_0)$; from the fact that the pairs in $list(i)$ are of the form (\ref{eqn:kDAG3}) or (\ref{eqn:kDAG4}), we have
\begin{equation}
\approx{s}(i,c) = \approx{s}(i,c_0) = \approx{s}(i_1,c_0 - w_{i_1}) \oplus \approx{s}(i_2,c_0 - w_{i_2}).
\label{eqn:kDAG5}
\end{equation}
There is no capacity $c_1$ in $list(i_1)$ such that $c_0-w_{i_1} < c_1 < c - w_{i_1}$. Indeed, for assuming the contrary, $c_1 + w_{i_1}$ would be a capacity in $list'(i)$, by~(\ref{eqn:kDAG3}). Since we have chosen $c_0$ as the largest capacity in $list(i)$ smaller than $c$, and $c_0 < c_1 + w_{i_1} < c$ holds, this implies that $c_1 + w_{i_1}$ was pruned when passing from $list'(i)$ to $list(i)$; thus, the two pairs of $list'(i)$ having $c_0$ and $c_1 + w_{i_1}$ as first components have equal second components. By (\ref{eqn:kDAG5}) and the bimonotonicity of $list(i_1)$, this entails that also the two pairs of $list(i_1)$ having $c_0-w_{i_1}$ and $c_1$ as first components must have equal second components. This contradicts the fact that $list(i_1)$ satisfies Property~\textsf{(I$_{1}$)}. Therefore, it holds that $\approx{s}(i_1,c_0 - w_{i_1}) = \approx{s}(i_1,c - w_{i_1})$. Analogously, we get $\approx{s}(i_2,c_0 - w_{i_2}) = \approx{s}(i_2,c - w_{i_2})$. 

Plugging these two relations into~(\ref{eqn:kDAG5}) we obtain 
\begin{equation}
\approx{s}(i,c) = \approx{s}(i_1,c - w_{i_1}) \oplus \approx{s}(i_2,c - w_{i_2}).
\label{eq:kk3}
\end{equation}
From~(\ref{eqn:kDAG2}), the fact that Property~\textsf{(I$_{2}$)} holds for $lists(i_1)$ and $lists(i_2)$, from~(\ref{eq:fl1}), and since $\ell(i) = 1 + \max\{\ell(i_1),\ell(i_2)\}$, the relation above implies that $\left(1-\eps/(n\log(\frac{m}{n}))\right)^{\ell(i)} s(i,c) \leq \approx{s}(i,c) \leq s(i,c)$, which shows that Property~\textsf{(I$_{2}$)} holds also for $list(i)$.
\end{proof}

We next prove an analogue of Lemma~\ref{lemma:list'}.

\begin{lemma}
\begin{sloppypar}
We can compute $list'(i)$ and $list(i)$ from $list(i_1)$ and $list(i_2)$ in time $O(n^2\log(\frac{m}{n})\eps^{-1}\lceil\log(1/\eps)/\log n\rceil)$.
\end{sloppypar}
\label{lemma:list'-kDAG}
\end{lemma}

\begin{proof}
At a generic step $i \inset{1}{n'}$, we compute $list'(i)$ as follows. We construct two auxiliary bimonotonic lists of pairs, $list_1(i)$ and $list_2(i)$. For every capacity $c_1$ in $list(i_1)$, $list_1(i)$ contains the pairs $[c_1+w_{i_1},\approx{s}(i_1,c_1) \oplus \approx{s}(i_2,c_1 + w_{i_1} - w_{i_2})]$. Analogously, for every capacity $c_2$ in $list(i_2)$, $list_2(i)$ contains the pairs $[c_2+w_{i_2},\approx{s}(i_1,c_2+w_{i_2}-w_{i_1}) \oplus \approx{s}(i_2,c_2)]$. List $list'(i)$ is now obtained by merging in a unique sorted list the lists $list_1(i)$ and $list_2(i)$. In order to prune the resulting bimonotonic list $list'(i)$ for obtaining $list(i)$, we do a linear scan.

In order to compute $list_1(i)$, proceed as follows (the computation of $list_2(i)$ is entirely analogous). Keep two pointers, $left$ in $list(i_1)$ and $right$ in $list(i_2)$. Pointer $left$ is initially set to the first pair in $list(i_1)$, say $[c_1,t_1]$. Pointer $right$ is set to the first pair in $list(i_2)$. If $c_1 + w_{i_1} - w_{i_2} < 0$, then we append the pair $[c_1+w_{i_1},t_1 ]$ at the end of $list_1(i)$. Otherwise, pointer $right$ starts scanning $list(i_2)$ until reaching a pair $[c_2,t_2]$, such that $c_1 + w_{i_1} - w_{i_2} \geq c_2$ and either $[c_2,t_2]$ is the last pair in $list(i_2)$, or $[c_2,t_2]$ is immediately followed by a pair $[c_3,t_3]$ with the property $c_1 + w_{i_1} - w_{i_2} < c_3$. Append the pair $[c_1+w_{i_1},t_1 \oplus t_2]$ at the end of $list_1(i)$. 

Afterwards, advance pointer $left$ to the next pair in $list(i_1)$, and repeat the above procedure, by advancing pointer $right$ to the corresponding pair, and inserting a new resulting pair in $list_1(i)$. This is repeated until pointer $left$ reaches the end of $list(i_1)$.

This completes the proof, since additions and comparisons on $O(\log C)$-bit numbers take unit time, floating-point additions and comparisons take $O(\lceil\log(1/\eps)/\log n\rceil)$ time, and the length of $list(i_1)$ and $list(i_2)$ is $O(n^2\log(\frac{m}{n})\eps^{-1})$.
\end{proof}

Thm.~\ref{thm:kDAG-main} follows from the facts that the transformed DAG $D'$ has $O(m)$ vertices, the length of the longest path in $D'$ is $O(n\log(\frac{m}{n}))$, and from Lemmas~\ref{lemma:list'-kDAG-0} and \ref{lemma:list'-kDAG}.

\end{document}